\documentclass[a4paper]{article}

\usepackage{graphics, color, latexsym, amssymb, amsmath, tabularx, epsfig, xspace} 
\usepackage{algorithm}
\usepackage[latin1]{inputenc}
\usepackage[T1]{fontenc}
\usepackage{amsfonts}
\usepackage{lmodern}
\usepackage{tikz}
\usepackage{cancel}
\usepackage{booktabs}
\usepackage{float}

\usepackage{enumitem}
\usepackage{amsthm}

\usepackage{geometry}
\usetikzlibrary{arrows,decorations,backgrounds}
\usetikzlibrary{patterns}
\usetikzlibrary{positioning}
\usetikzlibrary{shapes}
\usetikzlibrary{calc}

\usepackage{subcaption}
\usepackage{color}
\definecolor{TUGreen}{rgb}{0.517,0.721,0.094}
\definecolor{TUOrange}{rgb}{1.0,0.7176,0.0}
\definecolor{BrightGray}{gray}{0.9}
\definecolor{DarkGray}{gray}{0.2}
\definecolor{white}{rgb}{1,1,1}
\definecolor{black}{rgb}{0,0,0}
\definecolor{red}{rgb}{1,0,0}

\usepackage{xcolor}
\definecolor{itemizeblue}{rgb}{0.2,0.2,0.7}
\definecolor{shorse}{rgb}{0.84,0.84,0.94}
\definecolor{shorselight}{rgb}{0.93,0.93,0.98}
\definecolor{shorsedark}{rgb}{0.76,0.76,0.91}
\usetikzlibrary{topaths}
\usetikzlibrary{decorations.pathreplacing}

\usepackage{subcaption}
\newcommand{\dnode}[1]{\textnormal{d-node}(#1)}
\newcommand{\vars}[1]{\textnormal{vars}(#1)}
\newcommand{\rootset}[1]{\textnormal{root}(#1)}
\newcommand{\nodefunc}[1]{\textnormal{node}(#1)}

\newtheorem{proposition}{Proposition}
\newtheorem{theorem}{Theorem}
\newtheorem{lemma}{Lemma}
\newtheorem{corollary}{Corollary}
\newtheorem{simulation}{Simulation}

\theoremstyle{remark}
\newtheorem*{proofidea}{Proof idea}

\theoremstyle{definition}
\newtheorem{definition}{Definition}

\renewenvironment{proof}[1][Proof]{
                \begin{trivlist}
                \item \upshape \bfseries #1.
                \upshape\mdseries}
               {\nopagebreak[4]\hspace*{\fill}\mbox{$\Box$}\end{trivlist}}

\newtheorem{fact}{Fact}

\setitemize{itemsep=-2pt}
\setenumerate{itemsep=-2pt}

\begin{document}

\title{On the relation between structured $d$-DNNFs 
 and SDDs}

\author{Beate Bollig\thanks{TU Dortmund, LS2 Informatik, Germany, Email: {\tt beate.bollig@tu-dortmund.de}} \and Martin Farenholtz\thanks{Email: {\tt martin.farenholtz@tu-dortmund.de}}}

\maketitle

\begin{abstract}

Structured $d$-DNNFs and SDDs are restricted negation normal form circuits used in
knowledge compilation as target languages into which propositional theories are compiled.
Structuredness is imposed by so-called vtrees.
By definition SDDs are restricted structured $d$-DNNFs. 
Beame and Liew (2015) as well as Bova and Szeider (2017) mentioned the question 
whether structured $d$-DNNFs are really more general than SDDs w.r.t.\ polynomial-size
representations (w.r.t.\ the number of Boolean variables the represented functions are defined on.)
The main result in the paper is the proof that a function can be represented by SDDs of polynomial
size if the function and its complement have polynomial-size
structured $d$-DNNFs that respect the same vtree.


\end{abstract}

\section{Introduction}\label{sec1}

Knowledge compilation is an area of research with a long tradition in artificial
intelligence 
(see, e.g., 
\cite{CD97} and 
\cite{Mar15}).
A key aspect of any compilation approach is the target language 
into which the propositional theory is compiled. Therefore, 
the study of representations for propositional theories has been a central 
subject. One aim is to decide whether representations can be transformed 
into equivalent ones of another representation language by only increasing the 
representation size polynomially
(see, e.g., \cite{BLR13}, \cite{BL15}, \cite{BCM16}, \cite{BS17a}, \cite{Dar11}, \cite{DM02}, \cite{OD15}, \cite{Raz15}).
We follow this direction of research.

\paragraph{{\bf Representation languages 
}}

In the following we describe informally the significant representation languages considered in this paper
(for formal definitions of the significant ones see Section \ref{sec2}).
Circuits are a powerful model for the representation of Boolean functions in small size
(w.r.t.\ the number of Boolean variables the functions are defined on).
The desire to have 
better algorithmic properties leads to restricted circuits.
A circuit with $\wedge$-gates (conjunctions), $\vee$-gates (disjunctions), and 
negation gates is in negation normal form if its negation gates are only wired by input variables.
Decomposable negation normal form circuits, or DNNFs for short, 
introduced  by Darwiche \cite{Dar01} 
require that
the subcircuits leading into each $\wedge$-gate have to be defined on 
disjoint sets of variables. This requirement is called {\it decomposability}.
Darwiche also defined deterministic DNNFs, $d$-DNNFs for short,
where the subcircuits leading into each $\vee$-gate never simultaneously evaluate 
to the function value $1$. 
Determinism allows efficient algorithms for model counting, the efficient computation 
of the number of satisfying inputs 
of a function represented by a given restricted circuit. 
Moreover, Pipatsrisawat and Darwiche defined the notion 
of structured decomposability \cite{PD08}.
For every $\wedge$-gate the subcircuits leading into the gate are not only defined on 
disjoint sets of variables but their variables have to be splitted 
w.r.t.\ a given variable tree called vtree whose leaves are labeled by Boolean variables.
Structuredness
allows the combination of 
Boolean functions by a Boolean operator in 
polynomial time.
Sentential decision diagrams, SDDs for short, introduced by Darwiche \cite{Dar11}
are restricted structured $d$-DNNFs which 
have turned out to be a promising representation language for propositional knowledge bases 
as reported, e.g., by Van den Broeck and Darwiche \cite{BD15}. 
Here, the requirement of restricted strongly deterministic decompositions generalizes the well-known Shannon decomposition. 
It ensures that for each function $f$ representable by polynomial-size SDDs 
w.r.t.\ a vtree $T$ also the negated function
$\overline{f}$ is representable by SDDs w.r.t.\ $T$ in polynomial size.
DNNFs are not only by definition more general than $d$-DNNFs \cite{BCM16} and 
it is not difficult to prove that $d$-DNNFs are strictly more general than structured $d$-DNNFs.
Beame and Liew as well as Bova and Szeider mentioned the question about the relative succinctness 
of structured $d$-DNNFs and SDDs, i.e., 
whether structured $d$-DNNFs are strictly more concise than SDDs \cite{BL15,BS17}.
To the best of our knowledge it is still open whether there exists a Boolean function 
representable by structured $d$-DNNFs in polynomial size but needs SDDs of at least 
quasipolynomial or even exponential size.
In the following we tackle this question.

\paragraph{{\bf Contribution and related work}}

Ordered binary decision diagrams, 
OBDDs for short, are well suited as data structure for Boolean functions
and have received much considerations in the verification literature 
(see, e.g., the monograph of Wegener \cite{Weg00}).
OBDDs are restricted binary decison diagrams that respect  so-called variable orderings which are lists of variables.
SDDs respect variable trees and are more general than OBDDs by definition. 
Bova was the first one who presented a Boolean function representable by 
SDDs of polynomial size but whose OBDD size is exponential \cite{Bov16}.
Later on Bollig and Buttkus showed an exponential separation between SDDs and more general BDD models \cite{BB19}.
Unambiguous nondeterministic OBDDs have at most one accepting computation path for every input.
They can be seen as restricted structured $d$-DNNFs.
Recently, 
it was proved that every Boolean function
$f$ for which $f$ and its negated function $\overline{f}$ can be represented by
polynomial-size unambiguous nondeterministic OBDDs w.r.t.\ the same variable ordering can also be
represented by SDDs of polynomial size w.r.t.\ so-called linear vtrees that contain
additionally auxiliary variables \cite{BB19}. 
Here, we generalize this result and prove that 
given
polynomial-size structured $d$-DNNFs w.r.t.\ the same vtree for $f$ and 
its negated function $\overline{f}$, the function $f$  can also be
represented by SDDs of polynomial size. 
It is not difficult to prove that there are vtrees $T$ and Boolean functions $f$ 
such that structured $d$-DNNFs for $f$ and for $\overline{f}$ w.r.t.\ $T$ are
exponentially more succinct than SDDs representing $f$ w.r.t.\ $T$
(see also Section \ref{sec3}).
Therefore, for our simulation we have to modify the given vtree
and we add extra auxiliary variables.

\paragraph{{\bf Organization of the paper}}

The rest of the paper is organized as follows.
In Section \ref{sec2} we recall the main definitions concerning 
decomposable negation normal forms, we introduce certificates, and  
we investigate how structured $d$-DNNFs alter if some of the input variables are set to constants.
For completeness we prove in Section \ref{sec3} that structured $d$-DNNFs can be exponentially
more succinct than SDDs for a given function w.r.t.\ a fixed vtree.
Section 4 contains our main result.
It is shown that every Boolean function
$f$ for which $f$ and its negated function $\overline{f}$ can be represented by
polynomial-size structured $d$-DNNFs w.r.t.\ the same vtree can also be
represented by SDDs of polynomial size w.r.t.\ a vtree that contains extra auxiliary variables.

\section{Preliminaries}\label{sec2}

In this section, we briefly recall the main notions concerning 
decomposable negation normal forms, we take a look at certificates which are minimal satisfied subcircuits,
and we investigate how some representations of Boolean functions alter if some of the input variables
are replaced by Boolean constants. 

\subsection{\bf Decomposable negation normal forms}

We assume familiarity with fundamental concepts 
on Boolean functions and circuits (otherwise see, e.g., \cite{Vol99} for more details). 
In the rest of the paper, we look at (restricted) NNFs as classes of Boolean circuits.
A {\em satisfying input} for a Boolean function $f$ is an assignment to the input variables 
whose function value is $1$, in other words this assignment is mapped to $1$ by $f$.
A Boolean function $f$ depends essentially on a variable $x$ 
if the subfunctions of $f$ obtained by replacing $x$ by the Boolean
constants are different, in other words $f_{x=0}\not= f_{x=1}$.

Many known representations of propositional knowledge bases 
are restricted negation normal form circuits (NNFs)
and correspond to specific properties on NNFs \cite{DM02}. 
Decomposability and determinism are two of these fundamental properties.

\begin{definition}[NNFs]
A {\em negation normal form circuit} on a variable set 
$X$
is a Boolean circuit over fan-in 2 conjunction gates and unbounded fan-in disjunction gates, 
labeled by $\wedge$ and $\vee$,
whose inputs are labeled by literals 
$x$ and $\overline{x}$ for $x\in X$,
and $\bot$ and $\top$ for the Boolean constants $0$ and $1$.
The {\em size} of an \textup{NNF} $\mathcal{C}$, denoted by $|\mathcal{C}|$, is the number of its gates.
The {\em \textup{NNF} size} of a Boolean function $f$ is the size of a smallest
negation normal form circuit representing $f$.
The Boolean function $f_{\mathcal{C}}:\{0,1\}^X\rightarrow \{0,1\}$ 
represented by $\mathcal{C}$ is defined in the usual way.
An \textup{NNF} is {\em decomposable}, or a \textup{DNNF} for short, 
iff the children of each $\wedge$-gate are reachable from disjoint sets of input variables.
A set of Boolean functions $\{f_1, \ldots, f_\ell\}$ on the same variable set is {\em disjoint}
if each pair of functions $f_i,f_j$, $i\not= j$, 
is not simultaneously satisfiable.
A \textup{DNNF} is {\em deterministic}, or a $d$-\textup{DNNF} for short, iff the functions 
computed at the children of each $\vee$-gate are disjoint.
\end{definition}
The size of restricted NNFs for a Boolean function can be defined
in a similar way as above.
The fact that negations only appear at variables is not really a restriction and
the NNF size is polynomially related to the circuit size of a Boolean function over the standard basis 
$\{\wedge, \vee, \neg\}$.
Our assumption that each $\wedge$-gate has only fan-in $2$ is justified because it
affects the NNF size only polynomially. 
In the following for an \textup{NNF} $\mathcal{C}$ and a gate $g$ in $\mathcal{C}$ 
the subcircuit rooted at $g$ is denoted by $C_g$.

Structured decomposability on the notion of vtrees was originally introduced
by Pipatsrisawat and Darwiche \cite{PD08}. 
Informally, the variables are considered 
in a certain way formalized as a tree structure on the variables.

\begin{definition}\label{def:vtree}
A  {\em vtree} for a finite, nonempty set of variables $X$ 
is a full, rooted binary tree 
whose leaves are in one-to-one correspondence with the variables in $X$.
\end{definition}

In the rest of the paper for a node $u$ in a graph let vars$(u)$ denote the set of variables for which a literal
appears in the subgraph rooted at $u$.

\begin{definition}\label{definition:dnnf_respecting_vtree}
Let $T$ be a vtree for the finite, nonempty set of variables $X$ and $\mathcal{D}$
be a \textnormal{(deterministic) DNNF}.
An $\wedge$-gate $u$ of $\mathcal{D}$ with children $u_l, u_r$ respects a node
$v$ of $T$ with children $v_l, v_r$ iff
$\vars{u_l} \subseteq \vars{v_l}$ and $\vars{u_r} \subseteq \vars{v_r}$.
$\mathcal{D}$ \emph{respects} the vtree $T$, if every
$\wedge$-gate $u$ of $\mathcal{D}$ respects a node $v$ in $T$.
A \textup{(}deterministic\textup{)} $\textnormal{DNNF}$
that respects a given vtree $T$ is called a \textup{(}deterministic\textup{)}
$\textnormal{DNNF}_T$.
Moreover, a \emph{structured} \textup{(}deterministic\textup{)} $\textnormal{DNNF}$
is a \textup{(}deterministic\textup{)} $\textnormal{DNNF}_T$ for an arbitrary vtree $T$.
\end{definition}

In the following let $T_v$ be the subtree of a tree $T$ rooted at a node $v$ in $T$.
If we choose for every $\wedge$-gate $u$ in a DNNF$_T$ the node $v$ in $T$ in such a way
that $u$ respects $v$ and there exists no other node $v'$ in $T$ such that $u$
respects $v'$ and $T_{v'}$ is a subtree of $T_v$, the node $v$ is unique and
we call $v$ the \emph{decomposition node} of $u$
and $\dnode{u} = v$.

Sentential decision diagrams introduced by Darwiche \cite{Dar11}
result from so-called structured decomposability and strong determinism.
Therefore, by definition they are restricted structured $d$-DNNFs. 

\begin{definition}
For a variable set $X$ let $\bot:\{0,1\}^X\rightarrow \{0,1\}$ and
$\top:\{0,1\}^X\rightarrow \{0,1\}$ denote the constant $0$ function
and constant $1$ function, respectively.
A set of Boolean functions $\{f_1, \ldots, f_\ell\}$ on the same variable set is called a {\em partition}
iff the functions $f_1, \ldots, f_\ell$ are disjoint, none of the functions is the constant $0$ function $\bot$,
and $\bigvee\limits_{i=1}^\ell f_i=\top$.
\end{definition}

\begin{definition}\label{def:sdds}
A {\em sentential decision diagram} $C$,
or \textup{SDD} for short, {\em respecting a vtree} $T$ 
is defined inductively in the following way:
\begin{itemize}
\item $C$ consists of a single node labeled by a constant representing
        $\bot$ or $\top$ and $T$ is an arbitrary vtree.
\item $C$ consists of a single node labeled by a literal on a variable $x$
        and $x$ is in the variable set of $T$.
\item The output gate of $C$ is a disjunction whose inputs are wires from $\wedge$-gates $g_1, \ldots, g_\ell$,
      where each $g_i$ has wires from $p_i$ and $s_i$,
      $v$ is an inner node in $T$ with children $v_L$ and $v_R$,
      $C_{p_1}, \ldots, C_{p_\ell}$
       are \textup{SDDs} that respect the subtree of $T$ rooted at $v_L$,
      $C_{s_1}, \ldots, C_{s_\ell}$
      are \textup{SDDs} that respect the subtree of $T$ rooted at $v_R$,
      and the functions represented by $C_{p_1}, \ldots, C_{p_\ell}$ are a partition.
\end{itemize}
An SDD is an SDD respecting some vtree.
\end{definition}

By definition SDDs are circuits with alternating $\vee$- and $\wedge$-gates.
The property that the functions $p_1, \ldots, p_\ell$ are disjoint
is also called \emph{strong determinism} and the partition property is a strengthening of 
strong determinism. 
The distinction between the left and right child of a node is crucial for SDDs but not for structured $d$-DNNFs.
Xue, Choi, and Darwiche showed that switching the left and right child of a vtree node
may lead to an exponential change in the size of the corresponding SDDs for 
a given function \cite{XCD12}. Obviously, as a result we can conclude that for a function $f$
the size of a structured $d$-DNNF representing $f$ can be exponentially more succinct 
than its SDD size for a fixed vtree. In the next section we show that strong determinism is really 
a strengthening of determinism w.r.t.\ polynomial-size representations even without switching left and right
children in a vtree.

Some applications require their input circuits to be {\it smooth}, i.e., all children of an $\vee$-gate
mention the same set of variables. The notion of smoothness was first introduced by Darwiche \cite{Dar01a}. Recently,
smoothness was also considered for structured decomposable circuits \cite{SBB19}.

\begin{definition}
A circuit is {\em smooth} if for every pair of children $u_1$ and $u_2$ of an $\vee$-gate $u$, the set of variables 
$\vars{u_1}$ is equal to $\vars{u_2}$.
\end{definition}

A smoothing algorithm transforms a given circuit to a smooth circuit that represents the same Boolean function.
Since there exist smoothing algorithm for structured $d$-DNNFs that maintains the same vtree running in 
polynomial time \cite{SBB19}, we assume that the considered circuits are smooth in the rest of the paper.
For smooth structured $d$-DNNFs the notion of decomposition nodes can easily be extended for $\vee$-nodes.
Let $u$ be an $\vee$-node with children $u_1, \ldots, u_k$ then the decomposition node of $u$ is equal to the 
decomposition node of $u_1, \ldots, u_k$. 
Furthermore, decomposition nodes of nodes labeled by a literal are the corresponding leaves in the considered vtree,
respectively.

\subsection{\bf Certificates}

Bova et al  showed how to apply results from communication complexity to prove lower bounds on the size of (deterministic)
 DNNFs \cite{BCM16}. The main idea is that for a function $f$ this size is lower bounded by 
the size of a so-called balanced (disjoint) rectangle cover for $f$. For this result they considered certificates
for satisfying assignments w.r.t.\ a given (deterministic) DNNF $\mathcal{D}$. Informally,
these certificates are minimal satisfied subcircuits of $D$ that contain the output gate of $\mathcal{D}$.

\begin{definition}[\cite{BCM16}]
        Let $\mathcal{D}$ be a $\textnormal{DNNF}$ for the set of variables $X$. A \emph{certificate} of $\mathcal{D}$ is a $\textnormal{DNNF}$ $\mathcal{C}$ for $X$ with the following properties:
        \begin{itemize}
                \item[\textup{(i)}] The \textnormal{DNNF} $\mathcal{C}$ is a connected subgraph of $\mathcal{D}$.
                \item[\textup{(ii)}] The roots \textup{(}output gates\textup{)} of $\mathcal{C}$ and $\mathcal{D}$ coincide.
                \item[\textup{(iii)}] If $\mathcal{C}$ contains an $\wedge$-gate $u$, $\mathcal{C}$ also contains each child node $v$ of $u$ and the edge $(u,v)$.
                \item[\textup{(iv)}] If $\mathcal{C}$ contains an $\vee$-gate $u$, $\mathcal{C}$ also contains exactly one of the child nodes $v$ of $u$ and the edge $(u,v)$.
        \end{itemize}
\end{definition}

Since the fan-in of $\vee$-gates is restricted by $1$ and 
because of the decomposability of $\mathcal{D}$,
 a certificate can be seen as a rooted binary tree where each leaf is labeled by a different variable of $X$.
Now, we define $1$-certificates in order to represent sets of satisfying inputs of a given $\textnormal{DNNF}$
(for $1$-certificates see also \cite{BB19a}).

\begin{definition}
        A $1$\emph{-certificate} is a certificate with the following modifications: 
each leaf labeled by a literal $x$ is a decision node labeled by $x$ whose only outgoing edge labeled by $1$ 
leads to the $1$-sink and each leaf labeled by a literal $\overline{x}$ is a decision node labeled by $x$ 
whose only outgoing edge labeled by $0$ leads to the $1$-sink.
\end{definition}

A $1$-certificate represents all assignments to the input variables where the labels of outgoing edges of
decision nodes are chosen as assignments for the corresponding variables. Since a $1$-certificate does
not have to contain a decision node for each input variable, the represented set of assignments to
all variables 
can contain more than one element.
Obviously, according to the definition of $1$-certificates, each $\vee$- and $\wedge$-gate
evaluates to $1$ for an assignment of the represented set.
Since the roots of a $1$-certificate and a given $\textnormal{DNNF}$ coincide because of the second 
requirement in the definition of certificates,
this set of assignments is also satisfying for the given $\textnormal{DNNF}$.
In the deterministic case, there is a one-to-one correspondence between $1$-certificates and subsets 
of satisfying assignments for 
the function represented by the given DNNF.
We know the following fact \cite{BB19a}.
\begin{fact}\label{fact:certificates}
If $\mathcal{D}$ is a deterministic $\textnormal{DNNF}$ representing a Boolean function 
$f_{\mathcal{D}}: \{0,1\}^n \rightarrow \{0,1\}$, then 
for each satisfying assignment $b \in \{0,1\}^n$ of $f_{\mathcal{D}}$ 
there is a unique $1$-certificate of $\mathcal{D}$ representing $b$.
\end{fact}

\subsection{\bf Pruned vtrees and pruned structured $d$-DNNFs}

Beame and Liew showed how lower bounds on 
the SDD size for a Boolean function can be proved by
deterministic two-way best-partition communication complexity
\cite{BL15}.
For this reason they defined a small generalization of vtrees suitable for describing SDDs with 
respect to partial assignments. 

\begin{definition}[\cite{BL15}]
A \emph{pruned vtree} on a variable set $X$ is a full, rooted binary tree 
whose leaves are either marked \emph{stub} or by a variable in $X$
and whose leaves marked by variables are in one-to-one correspondence with the variables in $X$.
\end{definition}

Obviously, any vtree on a variable set $X$ is by definition also a pruned vtree on $X$.

For a vtree $T$ on a variable set $X$ and $A\subseteq X$ the pruned vtree $T_A$ can be constructed 
as follows. For each vertex $v$ in $T$, the subtree rooted at $v$ is replaced by a stub iff 
$\vars{v}\subseteq A$ and $\vars{\textup{parent}(v)} \not\subseteq A$, where parent$(v)$ is the unique 
parent of $v$ in $T$.

Next, we define pruned SDDs in a slightly modified form than Beame and Liew.

\begin{definition}
A {\em pruned} \textup{SDD} $C$ respecting a pruned vtree $T$
is defined inductively in the following way:
\begin{itemize}
\item $C$ consists of a single node labeled by a constant representing
        $\bot$ or $\top$ and $T$ is an arbitrary pruned vtree.
\item $C$ consists of a single node labeled by a literal on a variable $x$
      and $x$ is in the variable set of the pruned vtree $T$.
\item The output gate of $C$ is a disjunction whose inputs are wires 
      from $\wedge$-gates $g_1, \ldots, g_\ell$,
      where each $g_i$ has wires from $p_i$ and $s_i$,
      $v$ is an inner node in the pruned vtree $T$ with children $v_L$ and $v_R$,
      $C_{p_1}, \ldots, C_{p_\ell}$
       are pruned \textup{SDDs} that respect the subtree of $T$ rooted at $v_L$,
      $C_{s_1}, \ldots, C_{s_\ell}$
      are pruned \textup{SDDs} that respect the subtree of $T$ rooted at $v_R$,
      and the functions represented by $C_{p_1}, \ldots, C_{p_\ell}$ are mutually disjoint.
\end{itemize}

A {\em pruned SDD} is a pruned \textup{SDD} respecting some pruned vtree.
\end{definition}

Note that a pruned SDD w.r.t.\ a pruned vtree that is a stub represents always $\bot$ or $\top$.

Using pruned vtrees Beame and Liew investigated how SDDs simplify under partial 
assignments to the input variables. 
Let $ \mathcal{D}$ be an SDD representing a Boolean function $f$ w.r.t.\ the vtree $T$ on the set of variables $X$.
Furthermore, let $A\subseteq X$ and $p$ be 
an assignment to the variables in $A$. 
The pruned SDD $\mathcal{D}_p$ is constructed in the following way. 
For each leaf labeled by a literal $\ell$, where $\ell$ or $\overline{\ell}$ are in $A$,
we replace $\ell$ by $\top$ iff $\ell$ is fullfilled by $p$ and by $\bot$ otherwise.
Any children of $\vee$-gates that compute $\bot$ 
can be eliminated. If a gate computes a constant function under the assignment $p$, 
the outgoing edges of the gate can be replaced by the corresponding constant.
All nodes that are not any longer connected to the root of $\mathcal{D}$ can be eliminated.

It is not difficult to see that $\mathcal{D}_p$ is a subgraph of $\mathcal{D}$. Moreover, $\mathcal{D}_p$ is a
pruned SDD that respects the pruned vtree $T_A$ and $\mathcal{D}_p$ represents the subfunction of $f$ 
obtained by replacing the variables in $A$ by the assignment $p$ \cite{BL15}.
Pruned structured $d$-DNNFs can be obtained in a similar way as pruned SDDs. 

The following restricted partial assignments are crucial for our 
transformation from structured $d$-DNNFs into SDDs.

\begin{definition}[\cite{BL15}]
For a node $v$ in a vtree $T$ on the set of variables $X$ 
let $shell(v)$ denote the variables in $X\setminus \vars{v}$.

For $A\subseteq X$ we call $\{A, X\setminus A\}$ a \emph{shell partition} for $X$ if there is a vertex $v$ 
in the vtree $T$ such that $shell(v)=A$. An assignment to the variables in $A$ is called
a \emph{shell restriction} (w.r.t.\ $v$).
\end{definition}

Let $\mathcal{D}$ be a (pruned) 
SDD respecting a (pruned) 
vtree $T$ and let $v$ be a node in $T$. 
$SDD(\mathcal{D},v)$ denote the set of all the (pruned) SDDs in $\mathcal{D}$ that respect the (pruned)
subtree of $T$ 
rooted at $v$. In other words $SDD(\mathcal{D},v)$ contains
all (pruned) SDDs in $\mathcal{D}$ for which the decomposition node of the root is $v$.

\begin{proposition}[\cite{BL15}]\label{propos:sdds}
Let $\mathcal{D}$ be an \textup{SDD} respecting a vtree $T$ 
and let $v$ be a node in $T$. 
For every shell restriction $p$ w.r.t.\ $v$ the functions represented
by the pruned \textup{SDDs} in $SDD(\mathcal{D}_p,v)$ are mutually disjoint.
\end{proposition}

Here, we present a simpler alternative proof as the one of Beame and Liew \cite{BL15}.
\begin{proof}
The proof can easily be done by contradiction. Let $p$ be a shell restriction w.r.t.\ $v$ and let $f_1$ and $f_2$
be two functions represented by two different SDDs in $SDD(\mathcal{D}_p,v)$ that are not mutually disjoint.
Since the roots of the SDDs in $SDD(\mathcal{D}_p,v)$ have the same decomposition node, no SDD in
$SDD(\mathcal{D}_p,v)$ can be part of another SDD in $SDD(\mathcal{D}_p,v)$.
Let $u_i$ be the root of the SDD representing $f_i$ for $i\in \{1,2\}$.
Furthermore, let $r$ be an assignment to the variables in $\vars{v}$ which is mapped by $f_1$ and $f_2$ to $1$.
Now, let $p+r$ be the joint assignment to $X$ obtained from $p$ and $r$. 
It follows that there are two $1$-certificates for $p+r$ in the given SDD $\mathcal{D}$, 
one that contains $u_1$ and another one that contains $u_2$
 in contradiction to the determinism property
(see also Fact \ref{fact:certificates}).
\end{proof}

Strong determinism is not necessary for Proposition \ref{propos:sdds} 
because determinism is sufficient for Fact \ref{fact:certificates}. Therefore, we can generalize
Proposition \ref{propos:sdds} for structured $d$-DNNFs. Since the output gate of a structured $d$-DNNF 
can also be a conjunction, we have to modify the notation. 

\begin{definition}
Let $\mathcal{D}$ be a structured \textup{$d$-DNNF} w.r.t.\ a vtree $T$. 
Furthermore, let $v$ be an arbitrary node in $T$ and let $p$ be a shell restriction w.r.t.\ $v$, where 
$shell(v)=A$.
$\mathcal{D}_p$ denotes the pruned structured \textup{$d$-DNNF} obtained from $\mathcal{D}$ by the restriction $p$.
A node in $D_p$ is in the set $R(D_p,v)$ if its decomposition node in the pruned vtree $T_A$ is $v$.
Moreover, the set $R^+(\mathcal{D}_p,v)$ contains all nodes in $R(\mathcal{D}_p,v)$
for which there exists no child node in $R(\mathcal{D}_p,v)$.
The set $d$-$DNNF(\mathcal{D}_p,v)$ contains all pruned structured \textup{$d$-DNNFs} in $\mathcal{D}_p$
rooted at a node in $R^+(\mathcal{D}_p,v)$.
\end{definition}
Now, the following proposition can be proved in a similar way as Proposition \ref{propos:sdds}.

\begin{proposition}\label{propos:d-DNNF}
Let $\mathcal{D}$ be a structured $d$-\textup{DNNF} w.r.t.\ a vtree $T$ and $p$ 
be a shell restriction w.r.t.\ a node $v$ in $T$.
Then all functions represented by a pruned structured \textup{$d$-DNNF} in $d$-$DNNF(\mathcal{D}_p,v)$ 
are mutually disjoint.
\end{proposition}

The next observation is similar to the one for SDDs.
Let $\mathcal{D}$ be a structured $d$-DNNF w.r.t.\ a vtree $T$ and let $v$ be a node in $T$ with $shell(v)=A$. 
Furthermore, let $p$ be a shell restriction w.r.t.\ $v$.
If $f$ is the Boolean function represented by $\mathcal{D}$,
the disjunction of all functions represented by structured $d$-DNNFs in $d$-$DNNF(\mathcal{D}_p,v)$ 
is equal to the subfunction 
of $f$ obtained by replacing the variables in $A$ 
according to $p$.
In the following let $f_{A,p}$ denote this subfunction.

\begin{proposition}\label{propos:partition}
Let $\mathcal{D}$ be a structured $d$-\textup{DNNF} w.r.t.\ a vtree $T$ representing a Boolean function $f$
and $p$ be a shell restriction w.r.t.\ a node $v$ in $T$. Furthermore, let 
$\overline{ \mathcal{D}}$ be a structured $d$-\textup{DNNF} w.r.t.\ $T$ representing the Boolean function $\overline{f}$.
Then all functions represented by a pruned structured \textup{$d$-DNNF} in $d$-$DNNF(\mathcal{D}_p,v)$ or
in $d$-$DNNF(\overline{\mathcal{D}}_p,v)$ are mutually disjoint. Moreover, assuming that none of the functions is equal to the constant
function $\bot$, the set of all these functions is a partition.
\end{proposition}

\begin{proof}
We assume that none of the functions represented in 
$d$-$DNNF(\mathcal{D}_p,v)$ or $d$-$DNNF(\overline{\mathcal{D}}_p,v)$ is the constant function $\bot$. 
The disjunction of all the functions represented
in $d$-$DNNF(\mathcal{D}_p,v)$ 
($d$-$DNNF(\overline{\mathcal{D}}_p,v)$)  is equal to $f_{A,p}$ ($\overline{f}_{A,p}$) and
the disjunction of $f_{A,p}$ and $\overline{f}_{A,p}$ is obviously equal to the constant function $\top$.
Using Proposition \ref{propos:d-DNNF} it remains to prove that a pair of function $f_1$ and $f_2$ where
$f_1$ is represented by a structured $d$-DNNF in $d$-$DNNF(\mathcal{D}_p,v)$ and
$f_2$ is represented in $d$-$DNNF(\overline{\mathcal{D}}_p,v)$ is mutually disjoint. 
Since $f$ and $\overline{f}$ are mutually disjoint, we know that
$f_{A,p}$ and $\overline{f}_{A,p}$ are mutually disjoint. Therefore, we are done.
\end{proof}
\section{On Structured $d$-DNNFs and SDDs w.r.t.\ a Fixed Vtree}\label{sec3}

In this section, we show for completeness that strong determinism
is a stronger requirement than determinism w.r.t.\ polynomial-size representations.

\begin{proposition}
There exist a vtree $T$ and a Boolean function $f$ such that $f$ and $\overline{f}$ 
can be represented by structured \textup{$d$-DNNFs} w.r.t.\ $T$ in polynomial size but the \textup{SDD} size of $f$ 
w.r.t.\ $T$ is exponential.
\end{proposition}

\begin{proof}
We start with a well-known function from the BDD literature.
The {\it hidden weigthed bit function} HWB$_n$ introduced by Bryant \cite{Bry91}
is defined by 
$$
\textup{HWB}_n(x_1, \ldots, x_n)= x_{\|x\|},
$$
where $\|x\|= x_1 + \cdots + x_n$ is the number of variables set to $1$ in the input $x$ and
the output is $0$ if $x_1 + \cdots + x_n=0$.
W.l.o.g.\ let $n$ be divisible by $10$.
It is well known that the OBDD size of HWB$_n$ is $\Omega(2^{n/5})$ \cite{Bry91}.
Now, let $T$ be a vtree rooted at $v$ on the variables $x_1, \ldots, x_n$. The node $v_\ell$ is the left child of $v$,
and $v_r$ the right one. Furthermore, the subtree rooted at $v_\ell$ 
has $\frac{6}{10} n$ leaves and the one rooted at $v_r$ has $\frac{4}{10} n$ leaves. 
Both subtrees are right-linear which means that for every inner node in the considered vtree the left child is a leaf.
Let $X_\ell$ be the set of $x$-variables on leaves in the subtree rooted at $v_\ell$ and $X_r$ the remaining $x$-variables.

In the following we prove an exponential lower bound on the size of SDDs w.r.t.\ $T$ for HWB$_n$.
From the definition of SDDs it is not difficult to see that
it is sufficient to show an exponential 
lower bound  on the number of subfunctions for HWB$_n$ obtained by replacing 
$\frac{6}{10} n$ $x$-variables by constants. 
(For exponential lower bounds on the size of strongly deterministic structured $d$-DNNFs see also \cite{PD10}.)
Counting subfunctions is well-known from lower bound methods for OBDDs.
Therefore, similar lower bound proofs for the OBDD size of HWB$_n$ can be used.
A slight improvement of Bryant's lower bound is presented in Theorem in 4.10.2 in \cite{Weg00}.
It is proven that the number of subfunctions obtained by replacing 
$\frac{6}{10} n$ variables is at least $2^{\frac{n}{5}-1}$.

Next, we show that the structured $d$-DNNF size for HWB$_n$
w.r.t.\ $T$ is polynomial.
HWB$_n$ can be defined as $\bigvee_{1\leq i \leq n} E_i^n(X)\wedge x_i$,
where $E^n_i$ is the symmetric Boolean function on $n$
variables computing $1$ iff the number of ones in the input, that is the number of variables set to $1$,
is exactly $i$.
Now, we decompose the functions $(E_i^n(X)\wedge x_i)$ in the following way. 
The functions $f_{i,j}$ for $1\leq i\leq n$ and $0\leq i-\frac{4}{10}n \leq j \leq \frac{6}{10} n$ are defined
on the variables in $X_\ell$ and the function value of $f_{i,j}$ is $1$ iff the number of variables in $X_\ell$
set to $1$ is $j$ and $x_i$ is set to $1$ if $x_i\in X_\ell$. 
Furthermore, the functions $g_{i,j}$ for $1\leq i \leq n$ and $0\leq i-j\leq \frac{4}{10}n$ are defined on the variables in $X_r$. 
The function value of $g_{i,j}$ is $1$ iff the number of variables in $X_r$ set to $1$ is $i-j$ and
$x_i$ is set to $1$ if $x_i\in X_r$. 
Let $h_{i,j}(X)=f_{i,j}(X_\ell)\wedge g_{i,j}(X_r)$.
Then, for $i$ fixed $(E_i^n(X)\wedge x_i)$ is the disjunction of all function $h_{i,j}$ for 
$i-\frac{4}{10}n \leq j \leq \frac{6}{10} n$.
It is not difficult to see that each of the $f$- and $g$-functions can be represented by OBDDs 
of at most quadratic size w.r.t.\ every variable ordering because they are slight 
modifications of symmetric functions 
(see, e.g., \cite{Weg00}). Since OBDDs and SDDs w.r.t.\ right-linear vtrees are closely related (see 
\cite{Dar11} and Section 2.3 in \cite{BB19}), 
we can conclude that 
they can also be represented by polynomial-size structured $d$-DNNFs w.r.t.\ the right-linear
subtree rooted at $v_\ell$ and $v_r$, respectively. 
Obviously, HWB$_n$ is the disjunction of all $h_{i,j}$.
Neither the $f$-functions nor the $g$-functions are mutually disjoint but
two different functions $h_{i,j}$ and $h_{i',j'}$ for $i\not= i'$ or $j\not= j'$ are mutually disjoint. 

Since the number of the $f$- and $g$-functions as well as the number of 
conjunctions $f_{i,j}\wedge g_{i,j}$ 
is polynomially bounded in $n$ and the representation size for each $f$- and $g$-function
is polynomial, the $d$-DNNF size for HWB$_n$ w.r.t.\ $T$ is also polynomially bounded.

The polynomial upper bound on the structured $d$-DNNF size for $\overline{\text{HWB}}_n$ w.r.t.\ $T$ can 
be shown in a similar way using the representation

$$\overline{\textup{HWB}}_n(x)=\bigvee\limits_{1\leq k \leq n} (E^n_k(x)\wedge \overline{x}_k) \vee E^n_0(x).$$
Therefore, we are done.
\end{proof}

\section{Simulating Structured $d$-DNNFs by SDDs}\label{sec4}

In this section, we examine the relationship between stuctured $d$-DNNFs and SDDs.
We present a method to represent efficiently a Boolean function $f$ as an SDD 
provided that $f$ and $\overline{f}$ can both be represented by small structured $d$-DNNFs w.r.t.\
the same vtree. In Section \ref{sec3} we have seen that in general this is not 
possible without modifying the given vtree. In order to ensure the partition property of SDDs
we add some auxiliary variables. 
Our simulation generalizes the procedure described by Bollig and Buttkus
\cite{BB19} how to transform two unambiguous nondeterministic OBDDs w.r.t.\ the same 
variable ordering for $f$ and $\overline{f}$ into an SDD for $f$ in polynomial time.

\subsection{\bf Main ideas of the simulation}

Our simulation of structured $d$-DNNFs by SDDs can be divided into three phases.
First, we transform the given structured $d$-DNNF into an equivalent one
which is not only deterministic but strongly deterministic. 
For this reason we modify the given vtree into a vtree with auxiliary variables.
In the second phase we transform the resulting structured $d$-DNNF again 
into an equivalent one to ensure the partition property.
After the first and the second phase there can be edges between $\wedge$-gates 
but SDDs are by definition circuits with alternating $\vee$- and $\wedge$-gates. 
Therefore, in the last phase we transform subcircuits rooted at an $\wedge$-gate
into equivalent subcircuits rooted at an $\vee$-gate.

For the first and second phase we adapt ideas from the transformation of unambiguous 
nondeterministic OBDDs into SDDs \cite{BB19}.
Let $f$ and $\overline{f}$ be the Boolean functions represented by the structured $d$-DNNFs 
$\mathcal{D}$ and $\overline{\mathcal{D}}$ w.r.t.\ a vtree $T$.
Let $\mathcal{D}_u$ denote the subgraph of $\mathcal{D}$ rooted at a node $u$ 
and let $f_u$ be the Boolean function represented by $\mathcal{D}_u$.
Now, let $u$ be an arbitrary $\vee$-node in $\mathcal{D}$ and  
let $f_{u_1}, \dots, f_{u_l}$ be the functions represented at the child nodes of $u$.
Then $f_u$ is equal to
$(f_{u_1} \wedge \top) \vee (f_{u_2} \wedge \top) \vee \dots \vee (f_{u_l} \wedge \top)$.
Since $f_{u_1}, \dots, f_{u_l}$ and $\top$ formally have to be defined on disjoint sets of variables,
the given vtree has to be modified and auxiliary variables are added.
If every $\vee$-node in a given structured $d$-DNNF is handeled in this way, the result is 
an equivalent structured $d$-DNNF w.r.t.\ a modified vtree that is strongly deterministic.
(See Figure \ref{fig:strongdet} for a transformation of an $\vee$-gate.)

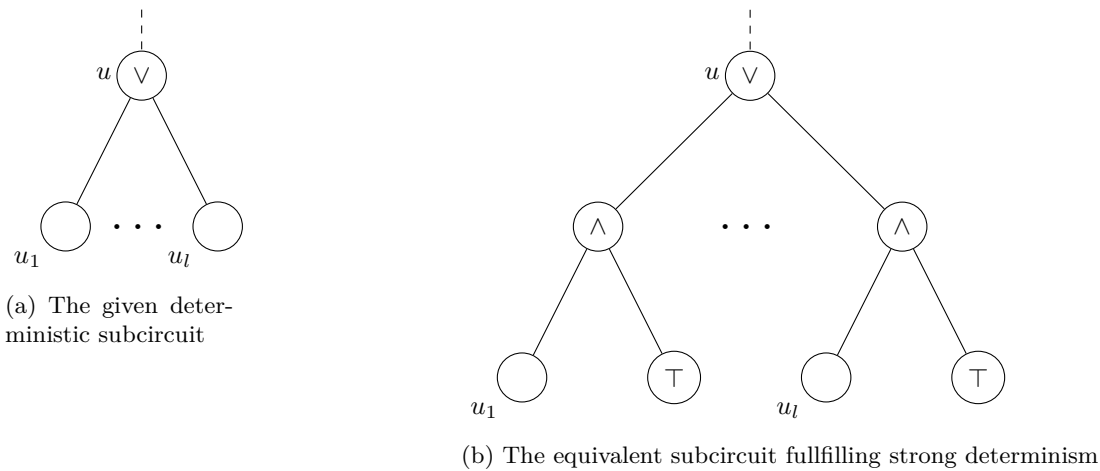
\begin{figure}[ht]
\begin{subfigure}[t]{0.2\textwidth}
	\vspace{0ex}
		\vspace{0ex}
		\begin{tikzpicture}[drawn/.style={circle,draw,minimum size =0.65cm}]
		\node (a) at (0,0) {};
		\node[drawn,label={[shift={(-0.5,-0.5)}]$u$}] (b) at (0,-1) {$\lor$};
		\node[drawn,label={[shift={(-.5,-1)}]$u_1$}] (c) at (-1,-3) {};
		\node[scale=1.75] at (0,-3) {$\dots$};
		\node[drawn,label={[shift={(-.5,-1)}]$u_l$}] (d) at (1,-3) {};
		\draw[dashed] (a) -- (b);
		\draw (b) -- (c);
		\draw (b) -- (d);
		\end{tikzpicture}
\subcaption{The given deterministic subcircuit}
\end{subfigure}\hfill
\begin{subfigure}[t]{0.6\textwidth}
	\vspace{0ex}
	\begin{tikzpicture}[drawn/.style={circle,draw,minimum size =0.65cm}]
	
	\node (start) at (6,2) {};
	\node[drawn,label={[shift={(-0.5,-0.5)}]$u$}] (a) at (6,1) {$\lor$};
	\draw[dashed] (start) -- (a);

	\node[drawn] (b2) at (4,-1) {$\land$};
	\node[drawn,label={[shift={(-.5,-1)}]$u_1$}] (c2) at (3,-3) {};
	\node[drawn] (d2) at (5,-3) {$\top$};
	\draw (a) -- (b2);
	\draw (b2) -- (c2);
	\draw (b2) -- (d2);	

	\node[scale=1.75] at (6,-1) {$\dots$};
	
	\node[drawn] (b3) at (8,-1) {$\land$};
	\node[drawn,label={[shift={(-.5,-1)}]$u_l$}] (c3) at (7,-3) {};
	\node[drawn] (d3) at (9,-3) {$\top$};
	\draw (a) -- (b3);
	\draw (b3) -- (c3);
	\draw (b3) -- (d3);

	\end{tikzpicture}
\subcaption{The equivalent subcircuit fullfilling strong determinism}
\end{subfigure}
	\caption{Transformation of an $\vee$-gate to ensure strong determinism}\label{fig:strongdet}
\end{figure}

The functions $f_{u_1}, \dots, f_{u_l}$ do not necessarily form a partition. Therefore,
the idea in the second phase is to find further functions represented at inner nodes of
$\mathcal{D}$ and $\overline{\mathcal{D}}$ which
together with $f_{u_1}, \dots, f_{u_l}$ yield a partition.
Let $A$ be a subset of the set of variables the function $f$ is defined on and let $p$ be 
an assignment of the variables in $A$.
Moreover, let $f_{A,p}$ denote the subfunction of $f$ obtained by replacing the variables in the subset $A$ by $p$.
The crucial observation for the second phase is that $f_{A,p}$ and $\overline{f}_{A,p}$ form a partition.
Now, let $v$ be the decomposition node of $u$ in the vtree $T$ and let the set $A$ be equal to $shell(v)$.
Furthermore, let $p$ be a shell restriction w.r.t.\ $v$ for which there
exists an assignment $r$ to the remaining variables such that there is a $1$-certificate for the joint assignment $p+r$
that contains the node $u$.
Then the disjunction of all functions represented by the pruned structured $d$-DNNFs in $d$-$DNNF(\mathcal{D}_p,v)$ 
is equal to $f_{A,p}$ and 
the disjunction of all functions represented by the pruned structured $d$-DNNFs 
in $d$-$DNNF(\overline{\mathcal{D}}_p,v)$ is equal to $\overline{f}_{A,p}$.
Moreover, all functions represented by a structured $d$-DNNF in 
$d$-$DNNF(\mathcal{D}_p,v)$ or $d$-$DNNF(\overline{\mathcal{D}}_p,v)$ are mutually disjoint 
(Proposition \ref{propos:partition}).
Now, 
let $u'_1, \ldots, u'_k$ be the set of all 
nodes in $R^{+}(\mathcal{D}_p,v)$ 
without $u_1, \ldots, u_l$ and 
let $w_1, \ldots, w_m$ be the set of all nodes in $R^{+}(\overline{\mathcal{D}}_p,v)$.
Then $f_{A,p}= f_{u_1} \vee \dots \vee f_{u_l} \vee f_{u_1'} \vee \dots \vee f_{u_k'}$
and $\overline{f}_{A,p}= \overline{f}_{w_1} \vee \dots \vee \overline{f}_{w_m}$.
Therefore, we can conclude that
\begin{eqnarray*}
        f_u &=& (f_{u_1} \wedge \top) \vee \dots \vee (f_{u_l} \wedge \top) \vee (f_{u_1'} \wedge \bot) \vee \dots \vee (f_{u_k'} \wedge \bot) \vee\\
        &&(\overline{f}_{w_1} \wedge \bot) \vee \dots \vee (\overline{f}_{w_m} \wedge \bot) 
\end{eqnarray*}
and 
since all the functions are mutually disjoint, they form a partition. 
(Here, we assume that none of these functions is equal to $\bot$.)

After the transformations in the first and second phase 
edges between $\wedge$-gates are possible in the resulting structured $d$-DNNF.
Since SDDs are by definition circuits with alternating $\vee$- and $\wedge$-gates, we 
modify some of the subcircuits rooted at $\wedge$-gates into equivalent ones rooted at $\vee$-gates.
Here, we have to make sure that the partition property is fullfilled.
The procedure is similar to the one in the second phase.
Let $u$ be such an $\wedge$-gate and $u_\ell$ and $u_r$ its child notes. Let $v_\ell$ be the decomposition node
of $u_\ell$ in the given vtree. 
Now, let $p$ be a shell restriction w.r.t.\ $v_\ell$ for which there 
exists an assignment $r$ to the remaining variables such that there is a $1$-certificate for the joint assignment $p+r$
that contains the node $u_\ell$. 
Let $u'_1, \ldots, u'_k$ be the set of all 
nodes in $R^{+}(\mathcal{D}_p,v)$ 
without $u_\ell$ or child notes of $u_\ell$
and let $w_1, \ldots, w_m$ be the 
set of all nodes in $R^{+}(\overline{\mathcal{D}}_p,v)$.
Then we know that 
\begin{eqnarray*}
        f_u &=& (f_{u_\ell} \wedge f_{u_r}) \vee 
         (f_{u_1'} \wedge \bot) \vee \dots \vee (f_{u_k'} \wedge \bot) \vee
(\overline{f}_{w_1} \wedge \bot) \vee \dots \vee (\overline{f}_{w_m} \wedge \bot) 
\end{eqnarray*}
and $f_{u_\ell}, f_{u_1'}, \ldots, f_{u'_k}, f_{w_1}, \ldots, f_{w_m}$ form a partition.
(See Figure \ref{fig:and} for a transformation of a function represented at an $\wedge$-gate.
For simplicity there are no nodes $u'_1, \ldots, u'_k$ in the figure.) 
Note that for the third phase we do not have to alter the already modified vtree.

\begin{figure}[ht]
	\begin{subfigure}[t]{0.2\textwidth}
		\vspace{0ex}
		\begin{tikzpicture}[drawn/.style={circle,draw,minimum size =0.65cm}]
			\node (a) at (0,0) {};
			\node[drawn,label={[shift={(-0.5,-0.5)}]$u$}] (b) at (0,-1) {$\land$};
			\node[drawn,label={[shift={(-.5,-1)}]$u_\ell$}] (c) at (-1,-3) {};
			\node[drawn,label={[shift={(-.5,-1)}]$u_r$}] (d) at (1,-3) {};
			\draw[dashed] (a) -- (b);
			\draw (b) -- (c);
			\draw (b) -- (d);
		\end{tikzpicture}
		\vfill
		\subcaption{The given subcircuit rooted at an $\wedge$-gate}
	\end{subfigure}
	\hfill
	\begin{subfigure}[t]{0.8\textwidth}
		\vspace{0ex}
		\begin{tikzpicture}[drawn/.style={circle,draw,minimum size =0.65cm}]
		
		\node (start) at (4,2) {};
		\node[drawn] (a) at (4,1) {$\lor$};
		\node[drawn] (b1) at (0,-1) {$\land$};
		\node[drawn,label={[shift={(-.5,-1)}]$u_\ell$}] (c1) at (-1,-3) {};
		\node[drawn,label={[shift={(-.5,-1)}]$u_r$}] (d1) at (1,-3) {};
		\draw[dashed] (start) -- (a);
		\draw (a) -- (b1);
                \draw (b1) -- (c1);
		\draw (b1) -- (d1);

		\node[drawn] (b2) at (3,-1) {$\land$};
		\node[drawn,label={[shift={(-.5,-1)}]$w_1$}] (c2) at (2,-3) {};
		\node[drawn] (d2) at (4,-3) {$\bot$};
		\draw (a) -- (b2);
		\draw (b2) -- (c2);
		\draw (b2) -- (d2);
		
		\node[scale=1.75] at (5,-1) {$\dots$};

		\node[drawn] (b3) at (7,-1) {$\land$};
		\node[drawn,label={[shift={(-.5,-1)}]$w_m$}] (c3) at (6,-3) {};
		\node[drawn] (d3) at (8,-3) {$\bot$};
		\draw (a) -- (b3);
		\draw (b3) -- (c3);
		\draw (b3) -- (d3);

		\end{tikzpicture}

		\subcaption{The equivalent subcircuit rooted at an $\vee$-gate}
	\end{subfigure}
	\caption{Transformation of a subcircuit rooted at an $\wedge$-gate}\label{fig:and}
\end{figure}
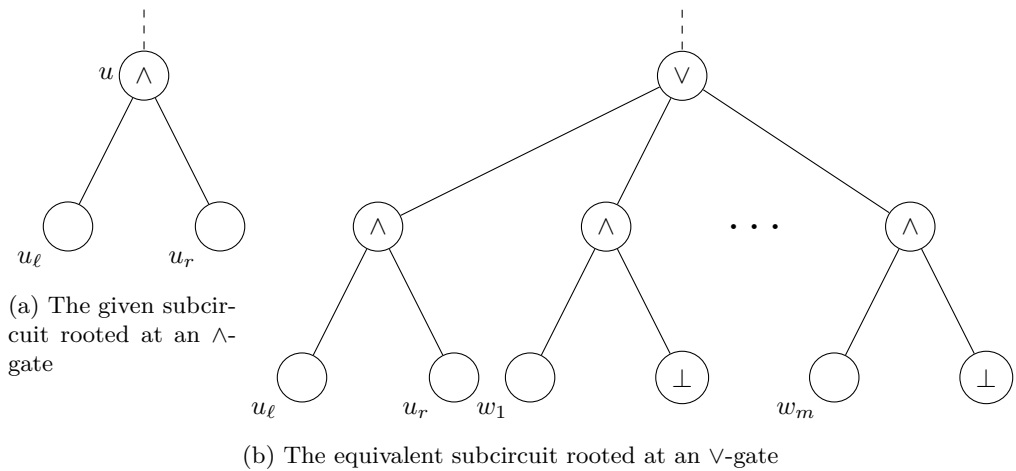

\subsection{\bf The Simulation}\label{subsec:simulation}

We start with a kind of normalization for vtrees. For our simulation later on it is not necessary that
the vtrees for the given structured $d$-DNNFs representing $f$ and $\overline{f}$ are the same 
but that they have the same normalized vtree.
Let $T$ be a given vtree. 
Since there is no difference between a left and right child 
of an inner node in a structured $d$-DNNF, we 
normalize the vtree in the sense 
that for each inner node 
$v$ the number of variables $\vars{v_\ell}$ of the left child of $v$ is at most as large as the number of variables
$\vars{v_r}$ of the right child of $v$. 
(See Figure \ref{fig:vtree} for an example of such a vtree.) The modified vtree $T'$ 
for our simulation has the following structure. 
For each inner node $v$ in $T$ for which $\vars{v}$ contains more than two variables
we add another node $v'$. If $v$ is the left (right) child of the node parent$(v)$, 
the parent of $v$ in $T$, $v'$ is the left (right) child of parent$(v)$ in $T'$.
Furthermore, $v'$ is the new parent of $v$ and $v$ is the left child of $v'$ in $T'$.
The right child of $v'$ is a leaf labeled by a new auxiliary variable. (See Figure \ref{fig:modvtree} for the 
modified vtree $T'$ w.r.t.\ the vtree $T$ in Figure \ref{fig:vtree}, 
a node $v$ in $T$ and the nodes $v$ and $v'$ in $T'$.) 
In the following $v$ denotes a node in $T$ as well as in $T'$. It will be clear from the context whether $v$
is in $T$ or $T'$.

	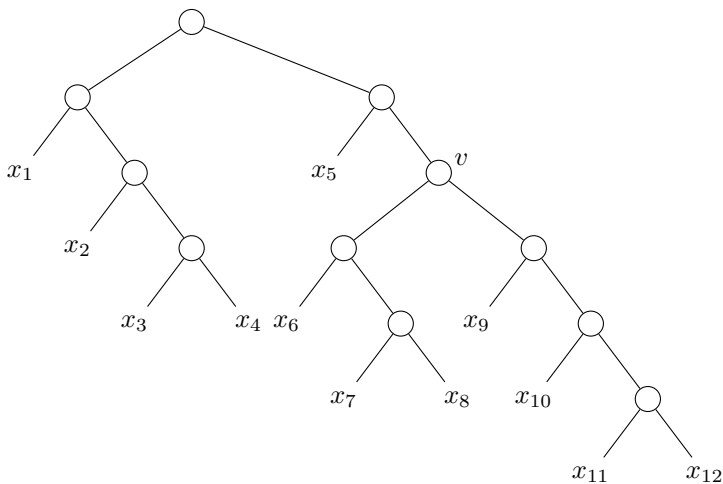
\begin{figure}[!ht]
 \begin{tikzpicture}[drawn/.style={circle,draw},level 1/.style={sibling
 	distance=30mm, level distance = 10mm},level 2/.style={sibling distance=15mm}]
 \node[drawn] {}
 child { node[drawn] {}
 	child { node {$x_1$}}
 	child { node[drawn] {}
 		child { node {$x_2$} }
 		child { node[drawn] {} 
 			child { node {$x_3$} }
 			child { node {$x_4$} }
 		}
 	}
 }
 child { node[drawn] at (1,0){}
 	child { node {$x_5$} }
 	child { node[drawn, label={[shift={(0.3,-0.2)}]$v$}] {}
 		child { node[drawn] at (-0.5,0){} 
 			child { node {$x_6$} }
 			child { node[drawn] {} 
 				child { node {$x_7$} }
 				child { node {$x_8$} }        
 			}        
 		}
 		child { node[drawn] at (0.5,0){} 
 			child { node {$x_9$} }
 			child { node[drawn] {}  
 				child { node {$x_{10}$} }
 				child { node[drawn] {} 
 					child { node {$x_{11}$} }
 					child { node {$x_{12}$} }        
 				}
 			}
 		}
 	}                
 };        
 \end{tikzpicture}
	\caption{A vtree $T$ on the Boolean variables $x_1, x_2, \ldots, x_{12}$}\label{fig:vtree}
\end{figure}

	\begin{figure}[!ht]
		\begin{tikzpicture}[drawn/.style={circle,draw},level 1/.style={sibling distance=25mm, level distance = 15mm},level 3/.style={sibling distance=20mm, level distance = 10mm}]
		\node[drawn] {}
		child { node[drawn] {}
			child { node[drawn] at (-1.5,-1.5){}
				child { node[drawn]{} 
					child { node {$x_1$} }
					child { node[drawn] {}
						child { node[drawn] {} 
							child { node {$x_2$} }
							child { node[drawn] at (0.4,-0.4){}
								child { node {$x_3$} }
								child { node {$x_4$} }	
							}
						}
						child { node {$h_1$} edge from parent[very thick]}
					}
				}
				child { node {$h_2$} edge from parent[very thick]}	
			}
			child { node[drawn] at (1.5,-1.5){}
				child { node[drawn] {} 
					child { node {$x_5$} }
					child { node[drawn, label={[shift={(0.4,-0.3)}]$v'$}] {}
						child { node[drawn, label={[shift={(0.4,-0.3)}]$v$}] {}
							child { node[drawn] at (-0.4,-0.4) {}
								child { node[drawn] {} 
									child { node {$x_6$} }
									child { node[drawn]  {}
										child { node {$x_7$} }
										child { node {$x_8$} }	
									}
								}
								child { node {$h_3$} edge from parent[very thick]}	
							}
							child { node[drawn] at (0.4,-0.4) {}
								child { node[drawn]  {}
									child { node {$x_9$} }
									child { node[drawn]  {}
										child { node[drawn]  {}
											child { node {$x_{10}$} }
											child { node[drawn]  {}
												child { node {$x_{11}$} }
												child { node {$x_{12}$} }	
											}	
										}
										child { node {$h_4$} edge from parent[very thick] }	
									}	
								}
								child { node  {$h_5$} edge from parent[very thick] }	
							}	
						}
						child { node {$h_6$} edge from parent[very thick]}	
					}
				}
				child { node {$h_7$} edge from parent[very thick]}
			}
		}
		child { node {$h_8$} edge from parent[very thick]};	
		\end{tikzpicture}
	
	\caption{The modified vtree w.r.t.\ the vtree $T$ in Figure \ref{fig:vtree}}\label{fig:modvtree}
	\end{figure}

For technical reasons we assume that the given structured $d$-DNNFs are of the following form.

\begin{definition}
Let $\mathcal{D}$ be a structured $d$-DNNF. We call $\mathcal{D}$ \em{simple} if
there exist no edges between $\vee$-gates, no $\vee$-gate is connected to a node labeled by $\bot$ or $\top$,
and nodes labeled by $\bot$ are the only ones in $\mathcal{D}$ that represent the constant function $\bot$.
\end{definition}

The size of a structured $d$-DNNF and the size of an equivalent simple one are polynomially related.
Moreover, 
it is easy to check whether a function represented by a $d$-DNNF
is equal to the constant function $\bot$ in polynomial time. 

In the followig simulation we look at DNNFs as graphs
and we start with a simple fact.

\begin{fact}\label{fact:twoVar}
Each Boolean function on two variables $x$ and $y$ can be represented by an \textup{SDD} w.r.t.\ a vtree on $x$ and $y$
with at most seven nodes.
\end{fact}

Such an SDD for a Boolean function can easily be constructed by evaluating each assignment to the variables.

\begin{simulation}\label{simulation:sdd}
Let $f$ be a Boolean function on $n$ variables and $f$ and $\overline{f}$ are represented by 
structured \textup{$d$-DNNFs} $\mathcal{D}$ and $\overline{\mathcal{D}}$ w.r.t.\ the same normalized
vtree $T$ on the set of variables $X$.
Because of Fact \ref{fact:twoVar} we assume that each function represented at 
a node whose decomposition node in $T$ is in a subtree $T_2$ on two variables is already 
represented by an \textup{SDD}
w.r.t.\ $T_2$.
Now, we construct an \textup{SDD} $S$ representing $f$ in the following way. 
First, we modify the vtree $T$ to a vtree $T'$ on the variables in $X$ and at most $n-2$ 
auxiliary variables as mentioned above.

Let $(V,E)$ and $(\overline{V},\overline{E})$ be the sets of nodes and edges of the
structured \textup{$d$-DNNFs} $\mathcal{D}$ and $\overline{\mathcal{D}}$, respectively.
Let $Y = V \cup \overline{V}$ and $Z = \{\emptyset, \wedge, \top, \bot\} \cup Y$.
The nodes of $S$ are tuples $(u,v) \in Y \times Z$.

We construct $S$ respecting the modified vtree by mapping nodes and edges of $\mathcal{D}$ and $\overline{\mathcal{D}}$
to nodes and edges of $S$ in the following way.
\begin{enumerate}
                \item For each node $u$ in $Y$ we add the node $(u,\emptyset)$ to $S$.
                      If the decomposition node of $u$ in $T$ is in a subtree of $T$ on two variables,
                      $(u,\emptyset)$ has the same label as $u$. 
                      If $(u,u')\in E$ or $(u,u')\in \overline{E}$, respectively, 
                      the edge $((u, \emptyset), (u', \emptyset))$ is inserted. 
                      Otherwise, if the decomposition node of $u$ in $T$ is not in a subtree of 
                      $T$ on two variables, $(u,\emptyset)$ is an $\vee$-node.

                \item For each $\vee$-node $u \in V$ whose decomposition node $v$ in $T$ is not in a subtree of $T$
                      on two variables the following nodes and edges are added.
                      Let $p$ be a shell restriction w.r.t.\ $v$ such that $u\in R(\mathcal{D}_p,v)$.
                      For each node $u' \in R^+(\mathcal{D}_p,v) \cup R^+(\overline{\mathcal{D}}_p,v)$, 
                      we add an $\wedge$-node $(u,u')$ to $S$.
                      Moreover, we add the nodes $(u,\bot)$ and $(u,\top)$ to $S$ 
                      which are labeled by $\bot$ and $\top$,
                      respectively. Furthermore, the edges 
                      $((u, \emptyset), (u, u'))$ and $((u, u'), (u', \emptyset))$
                      are added to $S$. 
                      In addition, if $(u,u')\in E$, the edge $((u, u'), (u, \top))$ is inserted, 
                      otherwise $((u, u'), (u, \bot))$.

                     In a similar way each $\vee$-node in $\overline{V}$ is handled.
              \item  For each $\wedge$-node $u\in V$ whose decomposition node $v$ in $T$ is not in a subtree of $T$
                     on two variables the following nodes and edges are added. 
                     Let $u_\ell$ and $u_r$ be the left and the right child of $u$,
                     respectively. Moreover, let $v_\ell$ be the decomposition node of $u_\ell$ in $T$ and 
                     let $p$ be a shell restriction w.r.t.\ $v_\ell$ 
                     such that $u_\ell\in R(\mathcal{D}_p,v_\ell)$.
                     We add the $\wedge$-node $(u,\wedge)$ and a node $(u,\bot)$ labeled $\bot$ and 
                     the edges $((u,\emptyset), (u, \wedge))$, 
                     $((u,\wedge),(u_\ell, \emptyset))$, $((u,\wedge),(u_r, \emptyset))$. 
                     Furthermore, for each node 
                     $u' \in R^+(\mathcal{D}_p,v_\ell) \cup R^+(\overline{\mathcal{D}}_p,v_\ell)$ without $u_\ell$ and 
                     any child node of $u_\ell$ we add the $\wedge$-node $(u,u')$ and the edges $((u,\emptyset),(u,u'))$,
                     $((u,u'), (u',\emptyset))$, $((u,u'),(u,\bot))$.

                      In a similar way each $\wedge$-node in $\overline{V}$ is handled.
\end{enumerate}

Let $\textnormal{root}(\mathcal{D})$ be the root of the structured \textup{$d$-DNNF} $\mathcal{D}$.
The root of $S$ is given by $(\textnormal{root}(\mathcal{D}), \emptyset)$. 
All nodes $(u,\top)$ and all nodes $(u,\bot)$ for $u\in Y$ can be merged to a node labeled $\top$ and  $\bot$, respectively. 
Finally, we remove all nodes and edges from the 
resulting \textnormal{SDD} $S$ which cannot be reached from 
$(\textnormal{root}(\mathcal{D}), \emptyset)$. 

\end{simulation}

	
\subsection{\bf Size, correctness, and equivalence}

In the following we show that the simulation presented in Subsection \ref{subsec:simulation}
can be done in polynomial-size. Furthermore, we prove that the result of the simulation
is an SDD for $f$, the Boolean function represented by one of the given structured $d$-DNNFs.
Our proofs generalize the proofs presented by Bollig and Buttkus \cite{BB19}.

We get a relationship between the sizes of the given structured $d$-DNNFs and the constructed equivalent SDD
by the following lemma which states that the increase in size is at most quadratic.

\begin{lemma}\label{lemma:size}
        Let $\mathcal{D}$ and $\overline{\mathcal{D}}$ be structured \textup{$d$-DNNFs}
        w.r.t.\ a normalized vtree $T$ that represent the Boolean functions $f$ and $\overline{f}$,
        respectively.
        Then, the result $S$ from Simulation \ref{simulation:sdd}
        has $\mathcal{O}((|\mathcal{D}|+|\mathcal{\overline{D}}|)^2)$ nodes and edges.
\end{lemma}

\begin{proof}
Since we assume in the paper that the given structured \textup{$d$-DNNFs} are smooth, we know that the size
of the vtree $T$ is not larger than the sizes of $\mathcal{D}$ or $\overline{\mathcal{D}}$. The modified vtree 
$T'$ is at most twice the size of $T$. Therefore, the number of nodes in $S$ is asymptotically
dominated by the size of $|Y\times Z|$ which is $\mathcal{O}((|\mathcal{D}|+|\mathcal{\overline{D}}|)^2)$.
The number of edges is linear in the number of nodes in $S$. Therefore, we are done.
\end{proof}

Each node $u$ in $V\cup \overline{V}$ is mapped to a node $(u,\emptyset)$ of $S$ in Simulation \ref{simulation:sdd}.
In order to prove that $S$ is a syntactically correct SDD that represents
the same function as $\mathcal{D}$, we prove that each subcircuit $S_{(u,\emptyset)}$ represents
the same function as $\mathcal{D}_u$ or $\overline{\mathcal{D}}_u$, respectively.
For this reason we map each node $u$ to a node in the modified vtree $T'$ which is 
the decomposition node of $(u,\emptyset)$. Therefore, 
$S_{(u,\emptyset)}$ respect the subtree of $T'$ rooted at this node.

Remember that for each node $v$ in the given vtree $T$ there is a node with the same name in $T'$.
Furthermore, if $\vars{v}$ is larger than two in $T$, the parent of $v$ in $T'$ is denoted by $v'$.

\begin{definition}
Let $T'$ be the modified vtree in Simulation \ref{simulation:sdd} and $u \in V \cup \overline{V}$ 
be a node of the given structured $d$-DNNFs.
The function \emph{node} maps nodes of $\mathcal{D}$ and $\overline{\mathcal{D}}$ 
to nodes of $T'$ in the following way.
\begin{eqnarray*}
  \textnormal{node}(u) :=
   \begin{cases}
    v, & \hspace{-0.25cm} v \textnormal{ is } \dnode{u} \textnormal{ in }T \textnormal{ and } \\
       & \hspace{-0.05cm}  u \textnormal{ is an } \wedge\textnormal{-node or }
               \vars{v} \textnormal{ is at most two in }T\\
   v', & \hspace{-0.25cm} v \textnormal{ is } \dnode{u} \textnormal{ in }T, \vars{v} 
                             \textnormal{ is larger than two,  and } \\
       & \hspace{-0.05cm}  u \textnormal{ is an } \vee\textnormal{-node}\\
  \end{cases}
\end{eqnarray*}
\end{definition}

\begin{lemma}\label{lemma:mainlemma}
Let $\mathcal{D}$ and $\overline{\mathcal{D}}$ be structured \textup{$d$-DNNFs}
w.r.t.\ a normalized vtree $T$ that represent the Boolean functions $f$ and $\overline{f}$. Furthermore,
let $S$ be the result from Simulation \ref{simulation:sdd} and $T'$
the modified vtree.
Then, each node $(u, \emptyset)$ of $S$ is the root of a syntactically correct \textnormal{SDD}
$S_{(u,\emptyset)}$ w.r.t.\ the vtree $T_{\textnormal{node}(u)}$.
Moreover, $S_{(u,\emptyset)}$ represents the same Boolean function as $\mathcal{D}_u$ for $u\in V$
or $\overline{\mathcal{D}}_u$ for $u\in \overline{V}$, respectively.
\end{lemma}

\begin{proofidea}
We give a proof by induction on the depth $l$ of the subgraph $S_{(u, \emptyset)}$ of the SDD $S$ in the appendix.
Here, the depth is the longest path to a leaf.
\end{proofidea}

As a result from Lemma \ref{lemma:mainlemma}, we know that $S$ is a syntactically correct SDD that represents 
the same function as the given structured $d$-DNNF $\mathcal{D}$.

\begin{corollary}\label{corollary:syntactic}
        Let $\mathcal{D}$ and $\overline{\mathcal{D}}$ be structured \textup{$d$-DNNFs}
        w.r.t.\ a normalized vtree $T$ that represent the Boolean functions $f$ and $\overline{f}$. Furthermore,
        let $S$ be the result from Simulation \ref{simulation:sdd}.
        Then, $S$ is a syntactically correct $\textnormal{SDD}$ w.r.t.\ the modified vtree $T'$ used in the simulation
        and $S$ represents $f$.
\end{corollary}

\begin{proof}
        The root of $S$ is given by the node $(\rootset{\mathcal{D}}, \emptyset)$ 
        described in Simulation \ref{simulation:sdd}. 
        Using Lemma \ref{lemma:mainlemma} we know that $S = S_{(\rootset{\mathcal{D}}, \emptyset)}$ 
        is a syntactically correct $\textnormal{SDD}$ w.r.t.\ the modified vtree $T'$ that represents the same
        Boolean function as the structured $d$-DNNF $\mathcal{D}$.
\end{proof}

\begin{theorem}\label{thm:transformation}
Let $f$ be a Boolean function such that $f$ and $\overline{f}$ can be represented by
structured \textup{$d$-DNNFs} $\mathcal{D}$ and $\overline{\mathcal{D}}$ w.r.t.\ the same normalized vtree.
Then, $f$ can also be represented by an \textnormal{SDD} of size
$\mathcal{O}((|\mathcal{D}|+|\overline{\mathcal{D}}|)^2)$.
\end{theorem}

\begin{proof}
Using Simulation \ref{simulation:sdd} we construct $S$ for $f$ w.r.t.\ the modified vtree $T'$
given the structured $d$-DNNFs 
$\mathcal{D}$ and $\overline{\mathcal{D}}$ representing the functions $f$ and $\overline{f}$ w.r.t.\ the same
normalized vtree $T$.
Using Corollary \ref{corollary:syntactic} we know that $S$ is a syntactically correct SDD representing $f$.
Moreover, from Lemma \ref{lemma:size} we know that the size of $S$ is at most quadratic w.r.t.\
the sizes of $\mathcal{D}$ and $\overline{\mathcal{D}}$.
\end{proof}

\section*{Concluding Remarks}

To the best of our knowledge the question whether the complexity class that consists
of all Boolean functions representable by polynomial-size structured $d$-DNNFs is closed 
under negation is open. We have shown how negation w.r.t.\ the same normalized vtree 
is useful to construct SDDs. For the construction we extended the given normalized vtree
to a vtree with further auxiliary variables. It is open whether a similar construction exists 
without any auxiliary variables. 


\bibliographystyle{spmpsci}  

\bibliography{references}

\appendix
\section*{Appendix: Proof of Lemma \ref{lemma:mainlemma}}

In the following proof we sometimes denote $S$ to be the Boolean function 
represented at the corresponding SDD $S$ by ease of notation. 
It will be clear from the context whether the 
SDD or the represented function is meant. Furthermore, $f_u$ and 
$\overline{f}_u$ denote the function represented by the subgraph of 
$\mathcal{D}$ or $\overline{\mathcal{D}}$, respectively, rooted at $u$.

\begin{proof}
We prove the lemma by induction on the depth $l$ of the subgraph $S_{(u, \emptyset)}$ of $S$.
Here, the depth is the length of a longest path to a leaf.

\textbf{Base case $(l \leq  2):$}

Since the depth of the subgraph $S_{(u,\emptyset)}$ is at most two 
and the given structured $d$-DNNFs are smooth,
$u$ has to be the root of an SDD w.r.t.\ a subtree of $T$ on at most two variables
in one of the given structured $d$-DNNFs.
Therefore, because of the first phase in Simulation \ref{simulation:sdd},
we can conclude that $(u, \emptyset)$ is the root of an SDD w.r.t.\ a subtree of $T'$ 
on at most two variables and $S_{(u,\emptyset)}$ represents the same function as 
$\mathcal{D}_u$ or $\overline{\mathcal{D}}_{u}$, respectively. Moreover, 
$\nodefunc{u}$ is $v$ in $T'$ if $\dnode{u}=v$.\\

\textbf{Induction hypothesis:}
Each subgraph $S_{(u,\emptyset)}$ of $S$ with depth of at most $l$
is a syntactically correct $\textnormal{SDD}$ w.r.t.\ the subtree $T'_{\textnormal{node}(u)}$.
Moreover, it represents the same Boolean function as 
$\mathcal{D}_u$ or $\overline{\mathcal{D}}_u$.\\

\textbf{Inductive step $(l \rightarrow l+1)$:}

\textbf{Case 1:} The node $(u, \emptyset)$ was added to $S$ because of the $\vee$-node 
                    $u \in V \cup \overline{V}$.

                    W.l.o.g.\ let $u \in V$. 
                    Let $v$ be $\dnode{u}$ in $T$ and 
                    let $p$ be a shell restriction w.r.t.\ $v$ such that $u\in R(\mathcal{D}_p,v)$.
                    The sets $R^+(\mathcal{D}_p,v)$ and $R^+(\overline{\mathcal{D}}_p,v)$ contain only
                    $\wedge$-nodes because of the definition of the sets, $\mathcal{D}$ 
                    and $\overline{\mathcal{D}}$ being simple, and the fact that $\vars{v}$ is at least three.
                    Therefore, $u$ is not in $R^+(\mathcal{D}_p,v)$ but the child nodes of $u$ are in the set.  

                    The node $(u, \emptyset)$ is an $\vee$-node which is connected to every $\wedge$-node 
                    $(u,u')$ for $u' \in R^+(\mathcal{D}_p,v) \cup R^+(\overline{\mathcal{D}}_p,v)$. 
                    These $\wedge$-nodes are connected to further nodes $(u',\emptyset)$ and $(u,\top)$ for
                    $(u,u')\in E$ or $(u,\bot)$ otherwise.
                    Thus, $S_{(u, \emptyset)}$ is 
                    inductively defined. Now, our aim is to show that 
                    $S_{(u, \emptyset)}$ is a syntactically correct SDD respecting the subtree 
                    rooted at $\nodefunc{u}$ in $T'$ that represents the same function as $\mathcal{D}_u$.
                    For this reason, we prove that the subgraphs rooted at $(u',\emptyset)$ in $S$
                    are syntactically correct SDDs w.r.t.\ the subtree $T'_{\nodefunc{u'}}$
                    which represent the same functions
                    as the ones represented by $\mathcal{D}_{u'}$ or $\overline{\mathcal{D}}_{u'}$, respectively.
                    Moreover, we have to show that they represent Boolean functions which form a partition.

                    The subgraph $S_{(u',\emptyset)}$ has at most depth $l-1$ for each 
                    $u' \in R^+(\mathcal{D}_p,v) \cup R^+(\overline{\mathcal{D}}_p,v)$ 
                    because by assumption $S_{(u, \emptyset)}$ is a subgraph of depth at most $l+1$ 
                    and $(u, \emptyset)$ is connected to each $(u', \emptyset)$ by a path of length two. 
                    Thus, by the use of the inductive hypothesis $S_{(u', \emptyset)}$ is a 
                    syntactically correct $\textnormal{SDD}$ w.r.t.\ the subtree rooted at 
                    $\nodefunc{u'}$ in $T'$.
                    We know that $\dnode{u'}=\dnode{u}=v$
                    since the given structured $d$-DNNFs are smooth. 
                    Moreover, $S_{(u', \emptyset)}$ represents the function $f_{u'}$
                    or $\overline{f}_{u'}$, respectively.

                    $S_{(u,\bot)}$ and $S_{(u, \top)}$ are $\textnormal{SDDs}$ 
                    representing $\bot$ and $\top$, respectively. Therefore, they respect the right subtree of 
                    $T_{v'}$, a leaf labeled by an auxiliary variable, where $v'$ is the parent of $v$
                    in $T'$. Therefore, $S_{(u, \emptyset)}$ respect $T_{v'}$. 

                    By induction hypothesis we know that $S_{(u', \emptyset)} = f_{u'}$
                    for each $u' \in R^+(\mathcal{D}_p,v)$ and 
                    $S_{(u', \emptyset)} = \overline{f}_{u'}$
                    for each $u' \in R^+(\overline{\mathcal{D}}_p,v)$. Using Proposition \ref{propos:partition}
                    and the fact that the given structured $d$-DNNFs are simple, we can conclude that all functions
                    $S_{(u', \emptyset)}$ for 
                    $u'\in R^+(\mathcal{D}_p,v) \cup R^+(\overline{\mathcal{D}}_p,v)$ form a partition.

Finally, we get the equivalence of $S_{(u, \emptyset)}$ and $f_u$ 
by applying the inductive hypothesis on $S_{(u', \emptyset)}$ for each  
$u' \in R^+(\mathcal{D}_p,v) \cup R^+(\overline{\mathcal{D}}_p,v)$.
\begin{eqnarray*}
S_{(u, \emptyset)}  &=& \bigvee_{\scriptsize{\substack{u' \,\in\, R^+(\mathcal{D}_p,v), \\ (u,u') \,\in\, E}}}
(S_{(u', \emptyset)} \wedge  S_{(u,\top)})  \,\vee\, \bigvee_{\scriptsize{\substack{u' \,\in\,  R^+(\mathcal{D}_p,v) 
\\ (u,u') \,\notin\, E}}}  (S_{(u', \emptyset)} \wedge  S_{(u,\bot)}) \,\vee\,
\bigvee_{\scriptsize{u' \,\in\,  R^+(\overline{\mathcal{D}}_p,v)}}
( S_{(u', \emptyset)}\wedge   S_{(u,\bot)})  \\
 &=& \bigvee_{\substack{u' \,\in\, R^+(\mathcal{D}_p,v), \\ (u,u') \,\in\, E}} 
       ( S_{(u', \emptyset)}  \wedge \top) \,\vee\, 
     \bigvee_{\substack{u' \,\in\, R^+(\mathcal{D}_p,v), 
\\ (u,u') \,\notin\, E}} ( S_{(u', \emptyset)}  \wedge \bot) \,\vee\,
     \bigvee_{u' \,\in\, R^+(\overline{\mathcal{D}}_p,v)}  ( S_{(u', \emptyset)}  \wedge \bot) \\
&=& \bigvee_{\substack{u' \,\in\, R^+(\mathcal{D}_p,v), \\ (u,u') \,\in\, E}} 
( S_{(u', \emptyset)}  \wedge \top)  = \bigvee_{(u,u') \,\in\, E} S_{(u', \emptyset)} 
\overset{\textnormal{(ind.)}}{=} \bigvee_{(u,u') \,\in\, E} 
f_{u'} \;=\; f_{u}
\end{eqnarray*}

\textbf{Case 2:} The node $(u, \emptyset)$ was added to $S$ because of the $\wedge$-node 
                    $u \in V \cup \overline{V}$.

                    W.l.o.g.\ let $u \in V$. 
                    Let $v$ be $\dnode{u}$ in $T$ and  $v_\ell$ be the left child of $v$
                    in $T$.
                    Furthermore, let $p$ be a shell restriction w.r.t.\ $v_\ell$ 
                    for which there exists an assignment $r'$
                    to the remaining variables such that there is a $1$-certificate in $\mathcal{D}$ 
                    for the joint assignment 
                    $p+r'$ to the $x$-variables that contains the node $u_\ell$.
                    Moreover, let $R^{++}(\mathcal{D}_p, v_\ell)$ contain all nodes
                    from $R^+(\mathcal{D}_p, v_\ell)$ without $u_\ell$ or any child node of $u_\ell$.

                    The node $(u,\emptyset)$ is connected to the $\wedge$-node $(u,\wedge)$ 
                    and to all $\wedge$-nodes $(u,u')$ for $u'\in 
                    R^{++}(\mathcal{D}_p, v_\ell)\cup R^+(\overline{\mathcal{D}}_p, v_\ell)$. The node $(u,\wedge)$
                    is connected to $(u_\ell, \emptyset)$ and $(u_r, \emptyset)$, the other $\wedge$-nodes
                    are connected to $(u',\emptyset)$ and $(u,\bot)$.

                    $S_{(u,\bot)}$ represents the function $\bot$ and is an SDD w.r.t.\ any vtree.
                    The subgraphs $S_{(u_\ell,\emptyset)}$, $S_{(u_r,\emptyset)}$, and
                    $S_{(u',\emptyset)}$ have at most depth $l-1$ 
                    because by assumption $S_{(u, \emptyset)}$ is a subgraph of depth at most $l+1$
                    and $(u, \emptyset)$ is connected to the nodes by paths of length two. 
                    Thus, by the use of the inductive hypothesis $S_{(u_\ell, \emptyset)}$,
                    $S_{(u_r,\emptyset)}$, and $S_{(u',\emptyset)}$ are 
                    syntactically correct $\textnormal{SDDs}$ w.r.t.\ subtrees rooted at 
                    $\nodefunc{u_\ell}$, $\nodefunc{u_r}$, and $\nodefunc{u'}$, respectively, in $T'$.
	            Moreover, $\nodefunc{u_\ell}$ and $\nodefunc{u'}$ are in the left subtree of $v$
                    and $\nodefunc{u_r}$ is in the right subtree of $v$.
                    The subgraph $S_{(u,\wedge)}$ represents the conjunction of the functions represented at
                    $S_{(u_\ell,\emptyset)}$ and $S_{(u_r,\emptyset)}$ and 
                    $S_{(u,u')}$ represents the conjunction of 
                    $S_{(u',\emptyset)}$ and $S_{(u,\bot)}$.  Therefore, 
                    $S_{(u,\wedge)}$ and $S_{(u',\emptyset)}$ respect $v$ in $T'$ and as a consequence also 
                    $S_{(u, \emptyset)}$ respect $v$.

                    If $u_\ell$ is an $\wedge$-node, $u_\ell\in R^+(\mathcal{D}_p,v_\ell)$. Therefore, 
                    $u_\ell \cup R^{++}(\mathcal{D}_p,v_\ell)$ is equal to 
                    $R^{+}(\mathcal{D}_p,v_\ell)$. We know that
                    $R^+(\mathcal{D}_p, v_\ell)\cup R^+(\overline{\mathcal{D}}_p, v_\ell)$ form a partition
                    because of Proposition \ref{propos:partition} and $\mathcal{D}$ being simple.
                    By the induction hypothesis the same holds for the functions represented by
                    $S_{(u_\ell,\emptyset)}$ and $S_{(u',\emptyset)}$ for $u'$ in 
                    $R^{++}(\mathcal{D}_p, v_\ell)\cup R^+(\overline{\mathcal{D}}_p, v_\ell)$.
                    If $u_\ell$ is an $\vee$-node, the function represented at $u_\ell$ is the disjunction
                    of the functions represented at the child nodes of $u_\ell$. 
                    Because of Proposition \ref{propos:d-DNNF} and $\mathcal{D}$ being simple
                    we know that all functions represented at a node 
                    in $\{u_\ell\} \cup R^{++}(\mathcal{D}_p, v_\ell)$ are mutually disjoint. The disjunction of these
                    functions is equal to $f_{A,p}$ where $A=shell(v_\ell)$. Now,
                    because of Proposition \ref{propos:partition} and $\mathcal{D}$ as well as 
                    $\overline{\mathcal{D}}$ being simple
                    we can conclude that
                    all functions represented at a node 
                    $\{u_\ell\} \cup R^{++}(\mathcal{D}_p, v_\ell)\cup R^+(\overline{\mathcal{D}}_p, v_\ell)$ 
                    form a partition.
                    By the induction hypothesis the same holds for the corresponding functions in $S$.

Finally, we get the equivalence of $S_{(u, \emptyset)}$ and $f_u$ 
by applying the inductive hypothesis on $S_{(u_\ell, \emptyset)}$ and $S_{(u_r, \emptyset)}$.
\begin{eqnarray*}
S_{(u, \emptyset)}  
&=& 
(S_{(u_\ell, \emptyset)} \wedge  S_{(u_r,\emptyset)})  \,\vee\, \bigvee_{u' \,\in\,  R^{++}(\mathcal{D}_p,v_\ell)}
(S_{(u', \emptyset)} \wedge  S_{(u,\bot)}) \,\vee\,
\bigvee_{u' \,\in\,  R^+(\overline{\mathcal{D}}_p,v_\ell)}
( S_{(u', \emptyset)}\wedge   S_{(u,\bot)})  \\
&=& (S_{(u_\ell, \emptyset)} \wedge  S_{(u_r,\emptyset)})  \,\vee\, \bigvee_{u' \,\in\,  R^{++}(\mathcal{D}_p,v_\ell)}
(S_{(u', \emptyset)} \wedge  \bot) \,\vee\,
\bigvee_{u' \,\in\,  R^+(\overline{\mathcal{D}}_p,v_\ell)}
( S_{(u', \emptyset)}\wedge   \bot))  \\
&=& (S_{(u_\ell, \emptyset)} \wedge  S_{(u_r,\emptyset)})  
\overset{\textnormal{(ind.)}}{=} 
f_{u_\ell} \wedge f_{u_r} \;=\; f_{u}
\end{eqnarray*}

\end{proof}

\end{document}